\documentclass[11pt]{elsarticle}
\usepackage{amscd,amssymb,amsmath,latexsym}

\usepackage{amsfonts}
\usepackage{graphicx}%
\usepackage[active]{srcltx}
\newtheorem{definition}{Definition}[section]

\newtheorem{rema}[definition]{Remark}
\newtheorem{exa}[definition]{Example}
\newtheorem{obs}[definition]{Observation}

\newtheorem{assu}[definition]{Assumption}

\newtheorem{lemma}[definition]{Lemma}
\newtheorem{proposition}[definition]{Proposition}
\newtheorem{theorem}[definition]{Theorem}
\newtheorem{corollary}[definition]{Corollary}

\newenvironment{remark}%
{\begin{rema}\rm}%
{\end{rema} }

\newenvironment{assumption}%
{\begin{assu}\rm}%
{\end{assu} }

\newenvironment{example}%
{\begin{exa}\rm}%
{\end{exa} }

{\begin{obs}\rm}%
{\end{obs} }

\newenvironment{proof}{\noindent{\bf Proof.}}{\hfill$\Box$}

\newcommand{\bfr}{{\boldsymbol{r}}}

\newcommand{\bfw}{{\boldsymbol{w}}}

\newcommand{\bfz}{{\boldsymbol{z}}}

\newcommand{\bfF}{{\boldsymbol{F}}}

\newcommand{\bfU}{{\boldsymbol{U}}}
\newcommand{\bfV}{{\boldsymbol{V}}}

\def\I{{\mathcal I}}

\def\C{{\mathbb C}}

\def\N{{\mathbb N}}
\def\Q{{\mathbb Q}}

\def\Z{{\mathbb Z}}

\newcommand{\ff}{{\bf f}}
\newcommand{\g}{{\bf g}}
\newcommand{\vv}{{\bf v}}

\newcommand{\sD}{\mathcal D}
\newcommand{\dnull}{\rm dnull}
\newcommand{\Null}{\rm null}

\begin{document}

\title{Certifying  solutions to overdetermined and singular polynomial systems  over $\Q$}

\author[NCSU]{Tulay Ayyildiz Akoglu \fnref{fnA}}
\author[ND]{Jonathan D. Hauenstein \fnref{fnH}}
\author[NCSU]{Agnes Szanto \fnref{fnS}}

\fntext[fnA]{This research was partly supported by a Turkish grant.}
\fntext[fnH]{This research was partly supported by DARPA YFA and NSF grant DMS-1262428.}
\fntext[fnS]{This research was partly supported by NSF grant CCF-1217557.}

\address[NCSU]{Department of Mathematics, North Carolina State University, Campus Box 8205, Raleigh, NC, 27965, USA.}
\address[ND]{Department of Applied and Computational Mathematics and Statistics, University of Notre Dame, Notre Dame, IN, 46556, USA.}

 \pagestyle{myheadings}
\markright{T. Akoglu, J.D. Hauenstein, A. Szanto }

\begin{abstract} This paper is concerned with certifying that a given point is near an exact root of an overdetermined or singular polynomial system with rational coefficients. The difficulty lies in the fact that consistency of overdetermined systems is not a continuous property. Our certification is based on hybrid symbolic-numeric methods to compute the exact {\em rational univariate representation} (RUR) of a component of the input system from approximate roots.  For overdetermined polynomial systems with simple roots, we compute an initial RUR from approximate roots.
The accuracy of the RUR is increased via Newton iterations 
until the exact RUR is found, which we certify using exact arithmetic.   Since the RUR is well-constrained, 
we can use it to certify the given approximate roots 
using $\alpha$-theory.
To certify isolated singular roots, we use a determinantal form of the
{\em  isosingular deflation}, 
which adds new polynomials to the original system
without introducing new variables. The resulting polynomial system is overdetermined, but the roots are now simple, thereby
reducing the problem to the overdetermined case. 
We prove that
our algorithms have complexity that are  polynomial in the input plus the output size upon successful convergence, and we use worst case upper bounds for termination when our iteration does not converge to an exact RUR. Examples are included to demonstrate the approach.\end{abstract}

\date{\today}

\maketitle

\section{Introduction}

In  their recent article \cite{Hauenstein-Sotille}, F. Sottile 
and the second author showed that one can get an efficient and practical root certification algorithm using  
\hbox{$\alpha$-theory} (c.f. \cite{Smale86,BlCuShSm}) for  well-constrained polynomials systems. The same paper also considers overdetermined systems over the rationals and show how to use 
$\alpha$-theory to certify that a given point is {\em not} an approximation of any exact roots of the system. However, to certify that a point is near an exact root, 
one can use universal lower bounds for the minimum of positive polynomials on a disk  That paper concludes that all known bounds
were ``too small to be practical.''

A closer look at the literature on lower bounds for the minimum of  positive polynomials over the roots of zero-dimensional rational polynomial systems reveals  that they all reduce the problem to the univariate case and use univariate root separation bounds (see, for example, \cite{Canny1990, Jeronimo-Perrucci,Bornawell-Yap,JePeTs2013}). This led to the idea of directly using an exact
univariate representation for certification
of the input system instead of using universal lower bounds 
that are often very pessimistic. 
For example, the~overdetermined~system 
$$f_1:=x_1-\frac{1}{2}, f_2:=x_2-x_1^2, \ldots, f_n:=x_n-x_{n-1}^2, f_{n+1}:=x_n
$$
has no common roots, but the value of $f_{n+1}$ on the common root of $f_1, \ldots, f_n$ is double exponentially small in $n$. While universal lower bounds cover these artificial cases, 
our approach, as we shall see, has the ability to terminate early in cases when the witness for our input instance is small. 

In principle, one can compute such a univariate representation using purely algebraic techniques, for example, by solving large linear systems corresponding to resultant or subresultant matrices (see, for example, \cite{Szanto2008}). However, this purely symbolic method would again lead to worst case complexity bounds. Instead, we propose a hybrid symbolic-numeric approach, using the approximate roots of the system, as well as exact univariate polynomial remaindering over $\Q$.  We expect that our method will make the certification of roots of overdetermined systems practical for cases when the universal lower bounds are too pessimistic, or when the actual size of our univariate representation is significantly smaller than in the worst case.

Consider an overdetermined system $\ff=(f_1, \ldots, f_m)\in \Q[x_1, \ldots, x_n]$ for some $m>n$, and assume that  the ideal  $\I:=\langle f_1, \ldots, f_m\rangle$ is radical and zero dimensional.  With these assumptions Roullier's Rational Univariate Representation (RUR) for  $\I$ exists  \cite{Rou99}, as well as for  any component of $\I$ over the rationals. Since the polynomials in the RUR have also rational coefficients, we can hope to compute them exactly, unlike the possibly irrational coordinates of the common roots of $\I$.  
With the exact RUR, which is a well-constrained system of polynomials, we can use $\alpha$-theory as in \cite{Hauenstein-Sotille} to certify that a given point is an approximate root for the RUR, and thus for our original system $\ff$. Alternatively, from an RUR, we can compute Hermite matrices for our system which can be used 
to certify that there is an exact root within $\epsilon$ 
of our given point.

Our numerical method to compute an RUR for $\I$ or for a rational component of $\I$ consists of three steps: 
\begin{enumerate}
\item compute approximations of all isolated roots 
to a given accuracy of a random well-constrained (square) set of linear combinations of the polynomials $f_1, \ldots, f_m$ using homotopy continuation (e.g., see \cite{BHSW13,SW05});
\item among the roots computed in Step 1, choose a subset that is a candidate to be approximations to roots of $\I$ or a rational component of $\I$ -- for this step, we can only give heuristics on how to proceed;
\item construct a rational  RUR from the approximate roots chosen in Step 2 using Lagrange interpolation and rational number reconstruction, and check whether the polynomials in $\ff$ reduce to zero modulo the computed RUR.  If yes, terminate, if not, continue to compute iteratively more accurate approximations of the RUR. 
\end{enumerate}

We propose two methods to increase the accuracy of the computed RUR. The first one, called {\em Local Newton Iteration}, simply repeats Step~3 with ever more accurate root approximations.  The second one,
called {\em Global Newton Iteration}, is a version of Hensel's lifting and was used in the non-Archimedian metric in, for example,
\cite{Kaltofen85b,Trinks85,Winkler88,Grobner-free}.  We apply it 
in the usual Euclidean metric on coefficient vectors with rational entries.   In~\cite{HausPanSza2014}, it was shown that the two methods are equivalent when using the p-adic non-Archimedian metric but different under the Euclidean metric.  
Furthermore,~\cite{HausPanSza2014}
gives sequential and parallel complexity analysis
for~the~iterations~of~both~of~these~methods.

Note that termination of the above steps depends
on the choice of approximate roots in Step 2.  
In particular, for wrong choices, the 
iterated RUR will never converge to an exact RUR of a rational component of $\I$. We will exhibit several 
approaches for termination: either by an {\em a priori} 
bound on the height of the coefficients of the exact RUR, 
or by incorporating certification of non-roots as 
in \cite{Hauenstein-Sotille}. 

Another approach that eliminates the heuristic choice in Step 2 
is to compute an exact RUR for all roots computed in Step 1.
From this RUR, using symbolic techniques, one can 
compute an exact RUR for $\I$. This method will always converge to an exact RUR, provided that the roots computed in Step 1 are certified approximate roots of the well-constrained system and the round off error is negligible.  Even though this technique is less efficient than the one described above, one may choose it if the method summarized above fails.  Note that this hybrid symbolic-numeric method  is highly parallelizable, which may make it preferable compared to purely symbolic methods to compute an RUR for $\I$ via elimination. 

In the second part of the paper, we consider certifying isolated singular roots of a rational polynomial system.  Due the behavior of Newton's method near singular roots
(e.g., see \cite{Griewank-Osborne}),
standard techniques in $\alpha$-theory can not be applied to certify such roots even if the polynomial system is well-constrained.
The key tool to handle such multiple roots is called {\em deflation}. 
Deflation techniques ``regularize'' the system thereby creating
a new polynomial system which has a simple root corresponding to the
multiple root of the original system \cite{DZ05,Hauenstein-Wampler,LVZ06,LVZ08,O87,OWM83}.  In this work,
we will focus on using a determinantal form of the
{\em isosingular deflation} \cite{Hauenstein-Wampler}, in 
which one simply adds new polynomials to the original system
without introducing new variables.  The new polynomials are constructed based on exact information that one can obtain 
from a numerical approximation of the multiple root.  
In particular, if the original system
had rational coefficients, the new polynomials which
remove the multiplicity information also have rational coefficients.
Thus, this technique has reduced us to the case of 
an overdetermined system over $\Q$ in the original set of variables that has a simple~root.

\subsection{Related work}

One of the applications of certifying near-exact  solutions of overdetermined and singular systems is the certification of computerized proofs of mathematical theorems, similarly as the 
software package {\tt alphaCertified} of \cite{Hauenstein-Sotille} was applied to confirm a conjecture of Littlewood in \cite{BozLeeRony2013}. 

As we mentioned in the Introduction, for well-constrained (square) polynomial systems, the paper  \cite{Hauenstein-Sotille} applies $\alpha$-theory (c.f. \cite{Smale86,BlCuShSm}) to obtain an efficient and practical root certifying  algorithm.  It is proposed to use universal lower bounds to certify approximate roots of overdetermined systems. In~\cite{DedieuShub2000}, an $\alpha$-theory for overdetermined systems was developed.  However, 
since this approach cannot distinguish between local minimums and roots, we can not use it to certify roots of overdetermined systems. There is an extended body of literature on using interval arithmetic and optimization methods to certify the existence or non-existence of the solutions of well-constrained systems with guaranteed accuracy, e.g.,
\cite{Krawczyk1969,Moore1977,Rump1983,YuKaTo1998,KanKashOish99,Naketal2003}.
Techniques to certify each step of path tracking in
homotopy continuation for well-constrained systems using 
$\alpha$-theory are presented in  \cite{BeltLey2012,BeltLey2013,HaHaLi2014}.  
Recently, a certification method of real roots of positive dimensional systems was studied in \cite{Yangetal2013}.
 
Related to the certification problems under consideration
is the problem of finding certified sum of squares (SOS) decompositions of rational polynomials. In \cite{PeyrlPar07,PeyrlParr2008,Kaltofenetal2008,Kaltofenetal2012},
they turn SOS decompositions given with approximate (floating point) coefficients into rational ones, assuming that the feasible domain of the corresponding semidefinite feasibility problem has nonempty interior. In \cite{MonnCor2011}, they adapt these techniques to the degenerate case, however they also require  a feasible solution with rational coefficients exists. The certification of more general polynomial, semi-algebraic and transcendental optimization problems were considered in \cite{AllamigeonGMW14}. In \cite{ElDinZhi2010}, they compute rational points in semi-algebraic sets and give a method to decide if a polynomial can be expressed as an SOS of rational polynomials.  Note that we can straightforwardly translate the certification of approximate roots of overdetermined polynomial systems into polynomial optimization problems over  compact convex sets (using a ball around the approximate 
root), however, we cannot guarantee a rational feasible solution. Instead, we propose to construct the rational representation of several irrational roots that form a rational component of the input system. Note also that the  coefficient vector  of the  RUR   of the input ideal  is the solution of a linear system  corresponding to multiples of the input polynomials that appears as the affine constraints in Lasserre's relaxation in
\cite{LassLauRos2007}, and thus can be computed using purely symbolic methods. However,   here we propose a more efficient symbolic-numeric approach,  constructing an RUR from approximate roots, and allowing RUR's of smaller rational components as well.  

The idea of computing an exact solution from numerical approximations is not new.  
In \cite{CaPaHaMo2001}, they give a randomized  algorithm that for a given approximate zero  $z$ corresponding to an exact zero $\xi$ of a polynomial system $F_1, \ldots, F_n\in \Q[x_1, \ldots, x_n]$ finds the RUR of the irreducible component $V_\xi\subset V(F_1, \ldots, F_n)$ in polynomial time depending on $n$, the heights of $F_i$, the degree of the minimal polynomial of a primitive element for $V_\xi$, and the height of this minimal polynomial. The algorithm in \cite{CaPaHaMo2001} uses the algorithm in \cite{KanLenLov1988} for the construction of the minimal polynomial of a given approximate algebraic number. However, the algorithm in \cite{KanLenLov1988} requires an upper bound for the height of the algebraic number, and this bound is used in the construction of the lattice that they apply LLL lattice basis reduction. To get such bound {\em a priori}, we would need to use universal bounds for the height.  In order to get  an incremental algorithm with early termination for the case when the output size is small,   one can modify the algorithm in \cite{KanLenLov1988} to be  incremental, but that would require multiple application of the lattice  basis reduction algorithm.  Alternatively,  one can apply the PSLQ algorithm as in \cite {FerBaiArn1999}, which is incremental and does not require an 
{\em a priori} height bound.   The main point of the approach in this paper is that we assume to know {\em all} approximate roots of a rational  component, so in this case we can compute the exact RUR much more efficiently, and in parallel, as we prove in this paper and in \cite{HausPanSza2014}.  So instead of multiple LLL lattice basis reduction, we propose a cheaper lifting and checking technique. 

Related literature on certification of singular zeroes of polynomial systems include \cite{KanOish97,RumpGrail2010,MantzMourr2011,LiZhi2013,LiZhi2014}. However, these approaches differ from ours in the sense that they certify singular roots of some  small perturbation of the input polynomial system while in the present  paper we certify singular roots of an exact polynomial system with rational coefficients.

\section{Certifying roots of overdetermined systems}\label{Sec:Over}

As described in the Introduction, the philosophy behind our method is to avoid using universal worst case lower bounds as certificates.  Instead, we aim to 
compute an exact univariate representation for our system that can be used for the certification. The Introduction outlined the steps 
of our proposed method to compute such a representation
with details~presented~here.

\subsection{Preliminaries}\label{Section:RUR}

Let us start with recalling the notion of Roullier's {\em Rational Univariate Representation (RUR)}, originally defined in \cite{Rou99}. We follow here the approach and notation in \cite{Grobner-free}.  

Let $\ff=(f_1, \ldots, f_m)\in \Q[x_1, \ldots, x_n]$ for some $m\geq n$, and assume that  the ideal  $\I:=\langle f_1, \ldots, f_m\rangle$ is radical and zero dimensional.  The factor ring $ \Q[x_1, \ldots, x_n]/ \I$ is a finite dimensional vector space over $\Q$, and we denote 
$$
\delta:=\dim_\Q \Q[x_1, \ldots, x_n]/\I
.$$
 Furthermore, for almost all $(\lambda_1, \ldots, \lambda_n)\in \Q^n$ (except a Zariski closed subset), the linear combination  
$$
u(x_1, \ldots, x_n):=\lambda_1 x_1 + \cdots + \lambda_nx_n
$$
is a {\em primitive element} of $\I$, i.e. the powers $1, u, u^2, \ldots, u^{\delta-1}$ form a linear basis for $\Q[x_1, \ldots, x_n]/ \I$ (c.f. \cite{Rou99}). 
Let $q(T)\in \Q[T]$ be the minimal polynomial of $u$ in  $\Q[x_1, \ldots, x_n]/ \I$, and let $x_i=v_i(u)$ be the polynomials expressing the coordinate function as  linear combinations of the powers of~$u$ in $\Q[x_1, \ldots, x_n]/ \I$. Note that from $u=\lambda_1 x_1 + \cdots + \lambda_nx_n$, we must have
$$
u=\lambda_1 v_1(u) + \cdots + \lambda_nv_n(u).
$$
and
$$
\langle
q(T), \; x_1-v_1(T), \ldots, x_n-v_n(T)\rangle= \langle \I, T-(\lambda_1 x_1 + \cdots + \lambda_nx_n)\rangle.
$$

\begin{definition} Let $\I=\langle f_1, \ldots, f_m\rangle \subset \Q[x_1, \ldots, x_n]$ be as above.  The  Rational Univariate Representation (RUR) of $\I$  is given by
\begin{itemize}
\item a primitive element $u=\lambda_1 x_1 + \cdots + \lambda_nx_n$ of $\I$ for some $\lambda_1, \ldots, \lambda_n\in \Q$;
\item the minimal polynomial $q(T)\in \Q[T]$ of $u$ in $\Q[x_1, \ldots, x_n] /I$, a monic  square-free polynomial of degree $\delta$;
\item the parametrization of the coordinates of the zeroes of $\I$ by the zeroes of $q$, given by
$$ v_1(T), \ldots, v_n(T)\in \Q[T]
$$
all degree at most $\delta-1$ and satisfying
$$
\lambda_1 v_1(T) + \cdots + \lambda_nv_n(T)\equiv T \mod q(T).
$$
\end{itemize}
\end{definition}

In the algorithms below, we compute an RUR that may not generate the same ideal $\I$ as our input polynomials, nevertheless 
it contains~$\I$, i.e., the polynomials $f_1, \ldots, f_m$ vanish modulo the RUR. In this case, the RUR will generate a {\em component of $\I$}. The common roots of the RUR of a component of $\I$ correspond to a subset of $V(\I)$.  For any subset $V\subset V(\I)$,
one can construct an RUR of the corresponding component 
of $\I$ that satisfies the following definition. 
Finally, we distinguish between RUR's of components 
and {\em rational components of $\I$}.

\begin{definition}  Let $\I=\langle f_1, \ldots, f_m\rangle\subset \Q[x_1, \ldots, x_n]$ be as above and
fix a primitive element $\lambda_1 x_1 + \cdots + \lambda_nx_n\in \Q[x_1, \ldots, x_n]$ of $\I$. 
The polynomials  
\begin{equation}\label{RUR}
T-\lambda_1 x_1 + \cdots + \lambda_nx_n, \;q(T),\; v_1(T), \ldots, v_n(T)
\end{equation}
form a {\em Rational Univariate Representation (RUR) of a component of $\I$} if it satisfies the following properties:
\begin{itemize}
\item $q(T)\in \C[T]$ is a monic square-free polynomial of degree $d\leq \delta$,
\item $v_1(T), \ldots, v_n(T)\in \C[T]$ are all degree at most $d-1$ and satisfy
$$
\lambda_1 v_1(T) + \cdots + \lambda_nv_n(T)\equiv T \mod q(T),
$$
\item for all $i=1, \ldots, m$ we have 
$$ f_i(v_1(T), \ldots, v_n(T)) \equiv 0 \mod q(T).
$$
\end{itemize}
If, in addition,  $q(T), v_1(T), \ldots, v_n(T)\in \Q[T]$ are rational polynomials, we call (\ref{RUR}) a {\em RUR of a rational component of $\I$}. 
\end{definition}

First note that the set
\begin{eqnarray}\label{GB}
\{ q(T), \; x_1-v_1(T), \ldots, x_n-v_n(T)\}
\end{eqnarray} 
forms a Gr\"obner basis for the ideal 
it generates with respect to the lexicographic monomial ordering defined by $T<x_1<\cdots < x_n$ and is well-constrained.  When we have an RUR of $\I$,  (\ref{GB}) is a Gr\"obner basis~for
$$
\langle \I, T-(\lambda_1 x_1 + \cdots + \lambda_nx_n)\rangle.
$$
In particular, one can compute an RUR of $\I$ with purely symbolic elimination methods, such as Buchberger's algorithm, or resultant based methods. 
Moreover, since (\ref{GB}) is well-constrained, one can apply
standard $\alpha$-theoretic tools to certify solutions.  

 Let 
 $$T-(\lambda_1 x_1 + \cdots + \lambda_nx_n),  \;q(T),\; v_1(T), \ldots, v_n(T)$$ be an RUR for $\I$ and 
 $$ \;q'(T),\; v'_1(T), \ldots, v'_n(T)\in \Q[T]$$ be an RUR for a rational component of $\I$ with respect to the same primitive element. Then, we have 
$$
q(T)\equiv 0   \text{ and } v_i(T)\equiv v_i' (T) \mod q'(T) \text{ for } i=1, \ldots, n.
$$
This shows that we can obtain an RUR for a rational component of $\I$ from an RUR of $\I$ using symbolic univariate polynomial factorization over $\Q$ and univariate polynomial remainders (we only use this in our complexity estimates, not in the computation that we propose). 
 
Next, let us recall the relationship between the RUR of  a component of~$\I$ and the corresponding (exact) roots.   Let  $V:=\{\xi_1, \ldots, \xi_d\}\subseteq V(\I)\subset \C^n$ be the exact roots of a component of $\I$, and denote $\xi_i=(\xi_{i,1}, \ldots, \xi_{i,n})$ for $i=1, \ldots, d$. Then for any $n$-tuple $(\lambda_1, \ldots \lambda_n)\in \Q^n$ such that 
$$
\lambda_1 \xi_{i,1} + \cdots + \lambda_n\xi_{i,n}\neq \lambda_1 \xi_{j,1} + \cdots + \lambda_n\xi_{j,n} \quad \text{ if } i\neq j,
$$
we can define a primitive element $u=\lambda_1 x_1 + \cdots + \lambda_nx_n$ for $V$. Since all roots are distinct, such primitive element exists and a randomly chosen 
\mbox{$n$-tuple} from a sufficiently large finite subset of $\Q^n$ will have this property with high probability (c.f. \cite{Rou99}). 
Fix such $(\lambda_1, \ldots \lambda_n)\in \Q^n$, and define 
\begin{eqnarray}\label{mui}
\mu_i:= \lambda_1 \xi_{i,1} + \cdots + \lambda_n\xi_{i,n}, \quad i=1, \ldots d.
\end{eqnarray}
Then,
\begin{eqnarray}\label{qT}
q(T):= \prod_{i=1}^d (T-\mu_i)
\end{eqnarray} is the unique monic polynomial of degree $d$ vanishing at the values of 
$u$ corresponding to each point in $V$. 
For each $j$, $v_j(T)$ in the 
parametrization of the coordinates  is the unique 
Lagrange~interpolant~satisfying
\begin{eqnarray}\label{vj}
v_j(\mu_i) = \xi_{i,j} \quad \text{ for } i=1, \ldots, d.
\end{eqnarray}

Unfortunately, the common roots of $\I$ may be irrational, so numerical methods will compute only approximations to them. However, the coefficients of the RUR  of $\I$ are rational numbers, so we may  compute them exactly. This is not always true for RUR's of components of $\I$, only for RUR's of rational components. Below, we will show how to recover the exact coefficients of the RUR of a rational component of $\I$ from numerical approximations of the roots of $\I$.  We start with giving heuristics for finding a good initial RUR which we will use as a starting point for 
iterative methods which are locally convergent to the exact RUR. 

\subsection{Initialization}\label{init}

As shown later, the iterative methods we use are {\em locally convergent}, so we need an initial RUR that is sufficiently close to the exact one in order 
to have convergence.  In this subsection, we discuss heuristics to find a good initial approximate RUR that we can use as the starting point for our iteration. 

%Furthermore, the local and global  Newton iterations we use below are only defined for well-constrained  systems with possibly more common roots than our overdetermined input or its rational components. Thus, we have to initialize  the subset of these roots based on the ones which are candidates for being roots of the component that we want to compute. 

More precisely, to compute an initial RUR for a rational component of~$\I$ we propose the following:
 
\begin{description}
\item[1. Homotopy method.]  Let $\ff=(f_1, \ldots, f_m)\in \Q[x_1, \ldots, x_n]$ for some $m>n$ be the defining equations of $\I$.  As in \cite[Section 3]{Hauenstein-Sotille}, for any linear map $R:\Q^m\rightarrow \Q^n$ which we will also consider as a matrix $R\in\Q^{n\times n}$, we define the well-constrained system $R(\ff):=R\circ f$. For almost all choices of $R\in \Q^{n\times m}$, the ideal generated by $R(\ff)$ is zero dimensional and radical.  We fix such an $R\in \Q^{n\times m}$ and throughout this paper we use the notation 
\begin{equation}\label{FF}
\bfF=(F_1, \ldots, F_n):= R(\ff).
\end{equation} 
In this step, we assume that by using numerical homotopy algorithms 
(c.f. \cite{BHSW13,SW05}), we have computed 
{\em approximate roots} for each root in  $V(\bfF)$, i.e.,  local Newton's method with respect to $\bfF$ is quadratically convergent starting from these approximate roots (see \cite[Section~2]{Hauenstein-Sotille} on how to certify approximate roots of well-constrained systems).  Using $\alpha$-theory for $\bfF$, we can estimate the distance
from each of these approximate roots to the exact ones. 

Denote an upper bound for these distances by $\varepsilon$. 

\item[2. Candidates for roots of a component of $\I$.] To find the subset of approximate roots to $V(\bfF)$ that approximates the roots in (a component of) $V(\I)$, we propose several methods. The first one is to choose the roots that has residuals for all $f_i$ for $i=1, \ldots, m$ up to a given tolerance $t$. The tolerance $t$ can be chosen based on $\varepsilon$ defined above in Step 1, and the height and degree of each of the polynomials in $\ff$.  Another approach could incorporate the ideas in \cite[Section~3]{Hauenstein-Sotille} to exclude the roots that are {\em not} approximations of $V(\I)$ by comparing the approximate roots of $R(\ff)$ to the approximate roots of an other random square subsystem $R'(\ff)$ for some $R'\in \Q^{n\times m}$. Finally, if we know that $\I$ or a rational component of $\I$ has only a very small number of common zeroes, then we can check  all subsets of the roots computed in Step 1 that have that cardinality as candidates for the common roots of a~component~of~$\I$. 

Denote the cardinality of the roots selected in this step by $d$. 

\item[3. Initial RUR.]  In this step, we compute an approximation to the RUR of a component of $\I$ from the $d$ candidates selected in the previous step using (\ref{mui}), (\ref{qT}) and (\ref{vj}). In Section~\ref{local}, we will discuss the sensitivity of this step to the perturbation of the roots, and give an upper bound for the distance of the approximate RUR computed from the~exact~one. 

\end{description}

As mentioned in the Introduction, we propose two iterative methods for increasng the precision of the approximate RUR: in the first we use {\em Local Newton Iteration} and the second we use {\em Global Newton Iteration}.  These are the subjects of the next two subsections.

\subsection{Increasing the Precision of the RUR using Local Newton Iteration}\label{local}

The main idea of increasing the precision of the RUR by local Newton iteration is very simple: we repeat the computation of   (\ref{mui}), (\ref{qT}), and (\ref{vj}) in the approximate RUR with ever more accurate roots that are computed by independently 
applying Newton's  iteration to each of the $d$~approximate~roots.

Here, we analyze the accuracy that we can achieve for the approximate RUR after $k$ iterations. 

Let $\bfF=(F_1, \ldots, F_n)$ be as in (\ref{FF}).  Let $\{\bfz^{(0)}_1, \ldots, \bfz^{(0)}_d\}\subset \C^n$ be the  the $d$ approximate roots we identified in Step 2 of Subsection \ref{init}, and let $\{\bfz^*_1, \ldots, \bfz^*_d\}\subset \C^n$ be the corresponding exact roots in $V(\bfF)$.   For each $i=1, \ldots, d$ and $k\geq 0$ we define the $(k+1)$-th Newton iterate by 
$$
 \bfz^{(k+1)}_i := \bfz^{(k)}_i - J_\bfF(\bfz^{(k)}_i )^{-1} F(\bfz^{(k)}_i ), 
 $$ 
 where $J_\bfF(\bfz^{(k)}_i )$ is the $n\times n$ Jacobian matrix of $\bfF$ evaluated at $\bfz^{(k)}_i$, which we assume to be invertible. Then, using our assumption in Step 1 of Subsection~\ref{init}, namely  that for all $i=1, \ldots, d$,
 $$
\| \bfz^{(0)}_i-\bfz^*_i\|_2\leq \varepsilon,
$$
and that the Newton iteration is quadratically convergent from each $\bfz^{(0)}_i$ to $\bfz^*_i$, we get that (using for example    \cite{BlCuShSm}) 
$$
\| \bfz^{(k)}_i-\bfz^*_i\|_2\leq \varepsilon\left(\frac{1}{2}\right)^{2^k-1}.
$$

Next, we analyze the possible loss of precision when applying (\ref{qT}) and~(\ref{vj}) in the computation of the approximate RUR.  Note that the Lagrange interpolation step in (\ref{vj}) may involve exponential loss of precision if we are not careful. In particular,   the  condition number of the Lagrange interpolant in the monomial basis is exponential in the number of nodes $d$ even if the nodes $\mu_i$ are in the interval $[-1, 1]$ (c.f.  \cite{Gautschi84}). However, using an orthogonal polynomial basis such as Chebyshev polynomials, the condition number of  the Lagrange interpolant in this basis  is linear in $d$, assuming that the nodes $\mu_i$ are in the interval $[-1,1]$  (it is $\sqrt{2}d$ for Chebyshev polynomials, c.f. \cite{Gautschi84}). Choosing an appropriate primitive element $u=\lambda_1 x_1 + \cdots + \lambda_nx_n$ ensures that all nodes $\mu_i$ are in $[-1,1]$.  Finally, we can convert the Chebyshev basis back to monomial bases by solving a triangular linear system with condition number at most $d2^{d-1}$ (c.f. \cite[Lemma 2]{GiesLabLee04}). 

Denote by $q^{(k)}(T)$ and $\vv^{(k)}(T)=(v_1^{(k)}(T), \ldots, v_n^{(k)}(T))$ the approximate RUR corresponding to $\{\bfz^{(k)}_1, \ldots, \bfz^{(k)}_d\}$, and  let $q^*(T)$ and $\vv^*(T)=(v_1^*(T), \ldots, v_n^*(T))$ be the exact RUR corresponding to $\{\bfz^*_1, \ldots, \bfz^*_d\}$. Then, using the above argument, we get the following bound for the error of the coefficients of the interpolation polynomials:
\begin{eqnarray}\label{errorbound}
\|v_j^{(k)}(T)-v_j^{*}(T)\|_2 \leq \varepsilon
d^2\left(\frac{1}{2}\right)^{2^k-1-d},
\end{eqnarray}
 which converges to zero as $k\rightarrow \infty$, since $d$ and $\varepsilon$ are fixed throughout our iteration.  We can use the same bound as in the right hand side of (\ref{errorbound}) for the accuracy of the polynomial $q^{(k)}(T)$. Note that to get the above bound we assume that there is no roundoff error in our computations, only the approximation error from the roots.
 
 \subsection{Increasing the Precision of the RUR using Global Newton Method}\label{global}

In this section, we give an adaptation of the global Newton method (also called multivariate Hensel lifting or Newton-Hensel lifting) from the non-Archimedean metric defined by a p-adic valuation used,
 for example, in~\cite{Kaltofen85,Trinks85,Winkler88,Grobner-free} to the Euclidean metric defined by the usual absolute value on $\Q$ or $\C$. Global Newton method increases the accuracy of the approximate RUR directly, using polynomial arithmetics, without using approximations of the roots, as shown below.
 
The recent work \cite{HausPanSza2014} includes several versions of the global Newton method that are shown to be equivalent when the coefficient ring possesses a p-adic absolute value, but different when considered over a coefficient field such as $\Q$ with the usual absolute value. One of these versions of the global Newton method is equivalent to that obtained from the local Newton iteration described above. Moreover, 
the parallel complexity of the different versions of the global Newton iteration
are compared in \cite{HausPanSza2014}, and demonstrated that they can be efficiently parallelized. 
 
 In order to make this paper self contained,  we recall a version of the global Newton method presented in \cite{HausPanSza2014} which is not equivalent to the one  in the previous  section, giving an alternative method to increase the accuracy. As we will see below, this version  relates to  some higher dimensional local Newton iteration. As a consequence, we prove the local quadratic convergence to the exact rational univariate representation with additional details provided in \cite{HausPanSza2014}. 
  
 Given $\bfF=(F_1, \ldots, F_n)$ and $u=\sum_{i=1}^n \lambda_ix_i$ in $\Q[x_1, \ldots, x_n]$ as before, we define the map $$\Phi:\Q^{(n+1)d}\rightarrow \Q^{(n+1)d}$$ as the map of the coefficient vectors of the following degree $d-1$ polynomials:
\begin{equation}\label{Phi}
\Phi : \;\;\left[\begin{array}{c}
v_1(T) \\ \vdots\\ v_n(T)\\ \Delta q(T) 
\end{array}\right] \mapsto \left[\begin{array}{c}
F_1(\vv(T))\hspace{-.3cm}\mod q(T)\\
 \vdots\\
F_n(\vv(T))\hspace{-.3cm}\mod q(T)\\
\sum_{i=1}^n \lambda_i v_i(T) -T
  \end{array}\right],
\end{equation}
where 
$$
 \Delta q(T):= q(T)-T^d.
 $$
If $u, q(T), v_1(T), \ldots, v_n(T)$ is an exact RUR of a component of $\langle \bfF\rangle$ then clearly
$$
\Phi\left(v_1(T), \ldots, v_n(T), \Delta q(T) \right)=0.
$$
We apply the $(n+1)d$ dimensional Newton iteration  to locally converge  to the coefficient vector  of an exact RUR which is a zero of $\Phi$. Note that below we will consider the map $\Phi$ as a map between 
$$
\Phi:\left( \Q[T]/\langle q(T) \rangle \right)^{n+1} \rightarrow \left( \Q[T]/\langle q(T) \rangle \right)^{n+1}, 
$$
and that $\left( \Q[T]/\langle q(T) \rangle \right)^{n+1}$ and $\Q^{(n+1)d}$ are isomorphic as vectors spaces. As was shown in  \cite{HausPanSza2014}, the Newton iteration for $\Phi$  respects the algebra structure of  $\left( \Q[T]/\langle q(T) \rangle \right)^{n+1}$ as well.

The first lemma gives the Jacobian matrix of $\Phi$.

\begin{lemma}[\cite{HausPanSza2014}] \label{PhiLemma} Let $\bfF=(F_1, \ldots, F_n)$, u, $q(T)$, $\vv(T)$ and $\Phi$ be as above.  For $i=1, \ldots, n$ define $m_i(T)$ and $r_i(T)$ as  the quotient  and remainder in the following division with remainder:
\begin{equation}\label{rm}
F_i(\vv(T))=m_i(T)q(T) + r_i(T).
\end{equation}
Then the Jacobian matrix of $\Phi$ defined in (\ref{Phi}) respects the algebra structure of  $\left( \Q[T]/\langle q(T) \rangle \right)^{n+1}$, and   is given by
\begin{equation}\label{JacobPhi}
J_\Phi( \vv(T) ,\Delta q(T))
:=
\begin{array}{|ccc|c|c}
\multicolumn{3}{c}{\scriptsize{n}}&\multicolumn{1}{c}{\scriptsize{1}}\\
\cline{1-4}
      & & & -m_1(T)&\\
      & J_\bfF(\vv(T)) & &\vdots& n\\
      & & & -m_n(T)&\\
      \cline{1-4}
   \lambda_1& \cdots &\lambda_n & 0 & 1  \\
\cline{1-4} \multicolumn{2}{c}{}
\end{array}\mod q(T).
\end{equation}
\end{lemma}

 The definition below is given using polynomial arithmetic modulo $q(T)$, but as we will see below, it is equivalent to the Newton iteration corresponding to $\Phi$ as a map defined on $\Q^{(n+1)d}$. For the definition to be well defined we need the following assumptions:
 
 \begin{assumption}\label{assu3} Let $\bfF$, $u=\sum_{i=1}^n \lambda_ix_i$, $q(T)$ and $\vv(T)$ polynomials over $\Q$  as above. We assume that 
\begin{enumerate}
\item $q(T)$ is monic and has degree $d$, 
\item $v_i(T)$ has degree at most $d-1$,

\item $\frac{\partial q(T)}{\partial T}$ is invertible modulo $q(T)$,
\item  $\lambda_1v_1(T)+\cdots +\lambda_n v_n(T) = T$,
\item    $J_\bfF(\vv(T))  := \left[\frac{\partial F_i}{\partial x_j}(\vv(T))\right]_{i,j=1}^n$  is invertible modulo $q(T)$.
\item $J_\Phi:=J_\Phi(\vv(T), \Delta q(T))$ defined in (\ref{JacobPhi}) is invertible modulo $q(T)$.
\end{enumerate}
\end{assumption}

\begin{definition}\label{constr3} If $\bfF(x_1, \ldots, x_n)$, $u(x_1, \ldots, x_n)$, $q(T)$ and $\vv(T)$ are polynomials over $\Q$ satisfying Assumption \ref{assu3}, we define 
{\small \begin{eqnarray*}
&&u=\sum_{i=1}^n\lambda_i x_i=T, \\
&&\bfw(T) :=\vv(T)-   \left(J_\bfF (\vv(T))^{-1}\bfF(\vv(T))\mod q(T)\right), \\
\end{eqnarray*}
\begin{eqnarray}
&&\Delta(T):=\sum_{i=1}^n \lambda_i w_i(T) -T,  \\
&& \bfr(T):= \bfF(\vv(T))\mod q(T),\label{r}\\
&& \bfU(T):=  \frac{\partial \vv(T)}{\partial T} -  \left(J_\bfF (\vv(T))^{-1}\frac{\partial \bfr(T)}{\partial T} \mod q(T)\right), \label{U}\\
&&\Lambda(T):= \sum_{i=1}^n \lambda_i U_i(T)\label{Lambda} \text{ that we will show to be invertible }\hspace{-.4cm}\mod q(T),\\
\label{bVV}&&{\bfV}(T):= \bfw(T)- \left( \frac{\Delta(T)}{\Lambda(T)} \bfU(T)\mod q(T)\right),\\
&&{Q}(T):=q(T) - \left(\frac{\Delta(T)}{\Lambda(T)}\frac{\partial q(T)}{\partial T} \mod q(T) \right).  \label{bQ}
\end{eqnarray}}
\end{definition}

The next proposition shows that ${\bfV}(T)$ and $ {Q}(T)$ from Definition \ref{constr3} are the Newton iterates for the function $\Phi$. 

\begin{proposition}[\cite{HausPanSza2014}]\label{mainprop}
Let $\bfF$, u, $q(T)$, $\vv(T)$ and $\Phi$ be such that  Assumption \ref{assu3} holds. Then, 
$\Lambda(T)$ defined in (\ref{Lambda}) is invertible modulo $q(T)$,   and thus ${\bfV}(T)$ and $ {Q}(T)$ are well  defined in Definition \ref{constr3}. Furthermore 
{\small \begin{equation}
\left[\begin{array}{c}{\bfV}(T)\\
 {Q}(T)-T^d\end{array}\right] = \left[\begin{array}{c}\vv(T)\\q(T)-T^d\end{array}\right]- J_\Phi^{-1}\cdot \left[\begin{array}{c}\bfF(\vv(T))\\\sum_{i=1}^n \lambda_i v_i(T)-T\end{array}\right] \mod q(T),
\end{equation}}
where the vector on the right hand side is $\Phi(\vv(T), q(T)-T^d)$. Finally, we also have that 
$$
\sum_{i=1}^n\lambda_i{V}_i(T) =T.
$$
\end{proposition}

As a corollary, we know that the approximate RUR defined by the iteration in Definition \ref{constr3} is locally quadratically
and converges to an exact RUR of a component of $\langle \bfF\rangle$. 

\begin{corollary}
The iteration defined by Definition \ref{constr3} is locally quadratically convergent to an exact RUR of a component of $\langle \bfF \rangle $, as long as Assumption \ref{assu3} is satisfied in each iteration.
\end{corollary}

In the next section, we will need estimates on the accuracy of our iterates after $k$ iterations of the global Newton method. Here, we can use $\alpha$-theory applied to the function $\Phi$ as described in 
\cite[Section 2]{Hauenstein-Sotille} to certify that the initial RUR computed in Subsection \ref{init} has the property of quadratic convergence and bound the distance to the exact RUR it converges to. 

Let $Q^{(0)}(T)$ and $\bfV^{(0)}(T)=(V_1^{(0)}(T), \ldots, V_n^{(0)}(T))$ denote our initial approximate RUR, and assume that it quadratically converges to the exact zero $Q^*(T)$ and 
$\bfV^*(T)=(V_1^*(T), \ldots, V_n^*(T))$ of the map $\Phi$. Note that this exact RUR may not be an RUR of a rational component of $\I$, i.e. it may not have rational coefficients.  At this stage, 
we will not be able decide its rationality.
Assume that we have a number $\nu$ such that
$$\| Q^{(0)}(T)-Q^{*}(T)\|_2\leq \nu , \quad \|V_i^{(0)}(T)- V_i^*(T)\|_2\leq \nu \quad i=1, \ldots, n, 
$$
where the norm is the usual Euclidean norm of the coefficient vectors of the polynomials over  $\C$.
Then, using the quadratic convergence of the global Newton iteration, according to \cite{BlCuShSm}, we get the following bound for the error of the coefficients of the polynomials in the $k$-th iteration, denoted by $Q^{(k)}(T)$ and $\bfV^{(k)}(T)=(V_1^{(k)}(T), \ldots, V_n^{(k)}(T))$:
\begin{eqnarray}\label{errorbound2}
\| Q^{(k)}(T)-Q^{*}(T)\|_2\leq \nu\left(\frac{ 1}{2}\right)^{2^k-1},\quad  \|V_i^{(k)}(T)-V_i^{*}(T)\|_2 \leq \nu\left(\frac{ 1}{2}\right)^{2^k-1}
\end{eqnarray}
which converge to zero as $k\rightarrow \infty$.  
 
\subsection{Rational Number Reconstruction}\label{rnr}

Below, we will show how to find the exact RUR of a rational component of $\I$ once we computed a sufficiently close approximation of it. The main idea is that we can reconstruct the unique rational numbers that have bounded denominators and indistinguishable from the coefficients of the polynomials in our approximate RUR within their accuracy estimates. Then, we can use purely symbolic methods to check whether the  RUR with the bounded rational coefficients is an exact RUR for a component of our input system~$\ff$. 
Here, we recall the theory behind rational number reconstruction, and in the next section we detail our ``end game,'' i.e. conditions for termination.
 
Since the coordinates of the approximate roots are given as floating point complex numbers, we can consider them as Gaussian rational in $\Q(i)$, and the same is true for the coefficients of the approximate RUR computed from these roots in Subsection \ref{local}. However, since the exact RUR has rational coefficients, we will neglect the imaginary part of the coefficients of the approximate RUR. Therefore, we will assume that the coefficients of the approximate RUR are 
in $\Q$, given as floating point numbers.

In this section, we use rational number reconstruction for each coefficient of the approximate RUR that was computed in the previous subsections. The following classical result implies that if a  number is sufficiently close to a rational number with small denominator, then we can find this latter rational number in polynomial time (c.f. \cite[Corrolary 6.3a]{Schrijver1986} or \cite[Theorem~5.26]{vzGathenGerhard1999}).

\begin{theorem}[\cite{Schrijver1986,vzGathenGerhard1999}] \label{ratrec} There exists a polynomial time algorithm which, for a given rational number $c$ and a natural number $B$ tests if there exists a pair of integers $(z,d)$ with $1\leq d\leq B$ and 
$$|c-z/d|<\frac{1}{2B^2},
$$ and if so, finds this unique pair of integers. 
\end{theorem}
To compute the pair $(z,d)\in \Z^2$ for each coefficient $c$ appearing in the approximate RUR computed in the previous subsections, we use the bound $B\in \N$ such that $\frac{1}{2B^2}\cong E$, where $E$ denotes our estimate of the accuracy of our approximate RUR either from (\ref{errorbound}) or from  (\ref{errorbound2}). Thus, we can define
$$
B:=\left\lceil (2E)^{-1/2}\right\rceil.
$$ 

For efficient computation of the rational number reconstruction, we can use the extended Euclidean algorithm or, equivalently, 
continued fractions (e.g., see \cite{vzGathenGerhard1999,WangPan2003,PanWang2004,Lichtblau2005,Steffy2010}), or use LLL basis reduction as in~\cite{BrightStorjohann2011}. The computed rational polynomials with bounded denominators~we~denote~by 
\begin{eqnarray}\label{hatRUR}\hat{q}(T)\text{ and }\hat{\vv}(T)=(\hat{v}_1(T), \ldots, \hat{v}_n(T)). 
\end{eqnarray}

\begin{remark} Theorem \ref{ratrec} does not guarantee existence of the the pair~$(z,d)$ with the given properties, only uniqueness. In case the rational number reconstruction algorithm for some coefficient $c$ returns that there is no rational number within distance $E$ with denominators at most $B$, we will need to improve the precision $E$ (which will increase the bound $B$ on the denominator). This is done by applying further local or global Newton steps on our approximates. As described in  Theorem \ref{worstcase} below, if the bound $B$ we obtained this way is larger than an 
{\em a priori} bound, we terminate our iteration and conclude that it did 
not converge to a RUR of a~rational~component.
\end{remark}

\subsection{Termination} 

One key task is to decide whether to terminate our iterations 
or increase the accuracy of the approximate RUR as described in Sections~\ref{local}~and~\ref{global}. 

Let $\hat{q}(T)\text{ and }\hat{\vv}(T)=(\hat{v}_1(T), \ldots, \hat{v}_n(T)) $ be the rational polynomials with bounded denominators 
as computed in  Subsection \ref{rnr}. In this step we will use our original overdetermined input system $\ff=(f_1, \ldots, f_m)$ in $\Q[x_1, \ldots, x_n]$ for some $m\geq n$ with $\I=\langle f_1, \ldots, f_m\rangle$. 

First, we reduce symbolically each $f_i$ by the Gr\"obner~basis 
$$\{\hat{q}(T), x_1-\hat{v}_1(T), \ldots, x_n-\hat{v}_n(T)\},
$$ 
or equivalently we compute $f_i(\hat{\vv}(T)) \mod q(T)$. If they all reduce to zero, we return (\ref{hatRUR}) as the  exact RUR of $\I$. 

If not all $f_i(\hat{\vv}(T)) \mod q(T)$ are zero, then either the accuracy of our approximate RUR was too small to ``click on'' the exact RUR in Subsection~\ref{rnr},  or the iteration does not converge to an RUR of a rational component of~$\I$. To decide which case we are in, we will use {\em a priori} upper bounds on the heights of the  coefficients of an RUR of a rational component~of~$\I$. 

Below, we review some of the known upper bounds that we can use in our estimates. First, we define heights of polynomials over $\Q$ in a way that we can utilize symbolic algorithms over $\Z$ to get bounds.  

\begin{definition}
Let $p(T)=T^d+a_{d-1}T^{d-1}+\cdots+a_0\in \Q[T]$ where each 
$a_i=z_i/d_i$ for $z_i\in \Z$ and $d_i\in \N$. Let $P(T)\in \Z[T]$ be any intereger polynomial that is an integer multiple of $p(T)$, for example we can choose  
$$P(T)=b_dT^d+b_{d-1}T^{d-1}+\cdots+b_0 :=\left( \prod_{i=0}^{d-1}d_i\right)p(T)\in \Z[T].
$$
Then the {\em height of $p$} is defined as
$$ H( p )=H( P )= \frac{\max \{|b_i|\;:\;i=0, \ldots , d\}}{\gcd(b_i\;:\; i=0, \ldots , d)}
$$
which is clearly independent of  representation of $p(T)$ in $\Z[T]$. Finally, we define the {\em logarithmic height of $p$} as
$$
h(p):=\log H(p).
$$
Note that if $\gcd(z_i, d_i)=1$ for all $i$ then 
\begin{equation}\label{denombound} H( p) \geq \max_i \{|z_i|, d_i\}
\end{equation}
so we can use the height to bound the magnitude of the numerators and denominators  appearing in the  coefficients of our polynomials. 
\end{definition}

The best known upper bounds for the logarithmic heights of the polynomials in the RUR of $\I$ are as follows \cite{DahanSchost2004}. Assume that the input polynomials $f_1, \ldots, f_m$ have degree at most $D$ and logarithmic height at most~$h$. First, we give a bound for the logarithmic bound of the primitive element $u= \lambda_1x_1+\cdots + \lambda_nx_n$ using \cite[Lemma 2.1]{Rou99}:
$$
h(\lambda_1x_1+\cdots + \lambda_nx_n)\leq 2(n-1)\log(n\delta),
$$ where $\delta$ is the number of roots in $V(\I)$. To use \cite{DahanSchost2004}, which assumes that $x_1$ is a primitive element, we take $(f_1, \ldots, f_m, \lambda_1x_1+\cdots + \lambda_nx_n - T)$ as our input, with a logarithmic height upper bound 
$$
h':=2(n-1)\log(n\delta)+h.
$$
Then, using an arithmetic B\'ezout
theorem in \cite[Lemma 2.7]{KrPaSo2001} and \cite[Theorem 1]{DahanSchost2004}, the logarithmic heights of the polynomials  $q(T), v_1(T), \ldots, v_n(T)$ in an RUR of $\I$ are bounded by 
 \begin{equation}\label{Hbound}
h(q), h(v_i)\leq 6n^3 h' D^{n+1}\leq 12 n^4h D^{n+1}\log(n\delta) \quad i=1, \dots n.
\end{equation}

To get a bound for the height of the polynomials in an RUR for a rational component of $\I$, we use Gelfund's inequality for the height of a polynomial divisor of an integer polynomial \cite[Proposition B.7.3]{HindSil2000} (other bounds can  also be used,  a survey of the known bounds of factors in $\Z[x]$ be found in~\cite{Abbott2013}). For $P,Q\in \Z[T]$ such that $P$ is a divisor of $ Q$, we have 
$$
H( P ) \leq e^{\deg(Q)} H(Q). 
$$
This, combined with (\ref{Hbound}), gives the following upper bound for the logarithmic heights of the polynomials in an exact RUR $q^*(T), v_1^*(T), \ldots, v_n^*(T)$ of a rational component of $\I$: 
\begin{equation}
h(q^*), h(v_i^*)\leq 12\delta n^4h D^{n+1}\log(n\delta)\quad i=1, \dots n, 
\end{equation}
where $\delta$ is the number of roots in   $V(\I)$.  This yields the following bound on the heights of the polynomials in a RUR 
of a rational~component~of~$\I$: 
\begin{equation}\label{Hbound2}
H(q^*), H(v_i^*)\leq Hn\delta e^{12\delta n^4 D^{n+1}}\quad i=1, \dots n, 
\end{equation}
where $H=e^h$ is an upper bound for the height of the input polynomials~in~$\ff$.
 
Once we have an {\em a priori} bound for the heights of the polynomials in an exact RUR of a rational component of $\I$, we can use (\ref{denombound}) and check if the bound $B$ for the denominators used in the rational number reconstruction in Subsection \ref{rnr} exceeds the {\em a priori} bound from (\ref{Hbound2}). If that is the case, we conclude that the iteration did not converge to an exact RUR of a rational component of $\I$ and terminate our algorithm. Otherwise, continue to increase the accuracy of our approximation. 

We summarize this subsection in the following theorems: 

\begin{theorem}\label{bestcase} Let $\I$ be as above. Assume that $q^*(T), v_1^*(T), \ldots, v_n^*(T)$ is an exact RUR of a rational component of $\I$. Define the maximum height
$$
H^*:=max\{H(q^*), H(v_1^*), \ldots, H(v_n^*)\}.
$$
Assume that an approximate RUR, $q(T), v_1(T), \ldots, v_n(T)$, satisfies 
\begin{equation}\label{2norm}
\|q(T)-q^*(T)\|_2, \|v_i(T)-v_i^*(T)\|_2\leq E < \frac{1}{2(H^*)^2} \quad i=1, \ldots, n,
\end{equation}
for some $E>0$, and let $\hat{q}(T), \hat{v}_1(T), \ldots, \hat{v}_n(T)$   obtained  via rational number reconstruction on the coefficients of $q(T), v_1(T), \ldots, v_n(T)$ using the bound $B:=\lceil(2E)^{-1/2}\rceil> H^*$. 
Then 
$$
\hat{q}(T)=q^*(T), \hat{v}_1(T)=v_1^*(T), \ldots, \hat{v}_n(T)=v_n^*(T).
$$
\end{theorem}

 \begin{proof} Note that the coefficients of 
 $q^*(T), v_1^*(T), \ldots, v_n^*(T)$ have denominator
 at most $H^*<B$ by (\ref{denombound}). Since the 2-norm gives an upper bound for the infinity norm, we know all coefficients of $q^*(T), v_1^*(T), \ldots, v_n^*(T)$ are at most distance $E$ from the corresponding coefficient of  $q(T), v_1(T), \ldots, v_n(T)$. By Theorem~\ref{ratrec}, for each coefficient  of of  $q(T), v_1(T), \ldots, v_n(T)$, there is at most one rational number   with denominator bounded by \mbox{$B=\left\lceil (2E)^{-1/2}\right\rceil> H^*$} within the distance of  
 $$
\frac{1}{2B^2}= \frac{1}{2\left\lceil (2E)^{-1/2}\right\rceil^2}\leq \frac{1}{2\left( (2E)^{-1/2}\right)^2}= E.
$$
This proves that  the rational reconstruction must equal the exact~RUR.  \end{proof}
 
 The next theorem considers the converse statement.

\begin{theorem}\label{worstcase} Let $\ff\subset \Q[x_1, \ldots, x_n]$ be as above and assume that $H$ and~$D$ are the maximum height and degree  of the polynomials in $\ff$ respectively. Also, let $\delta$ be the cardinality of the common roots of $\ff$ in $\C^n$. Assume that we have an upper bound $E$ for the accuracy of our approximate RUR $q(T), v_1(T), \ldots, v_n(T)$ either from (\ref{errorbound}) or from  (\ref{errorbound2}). Let $B:= \lceil(2E)^{-1/2}\rceil$ and assume that 
$$
B\geq  Hn\delta e^{12\delta n^4 D^{n+1}}.
$$
Let $\hat{q}(T), \hat{v}_1(T), \ldots, \hat{v}_n(T)$ be   obtained  via rational number reconstruction from the coefficients of $q(T), v_1(T), \ldots, v_n(T)$ using the bound  $B$ for the denominators. If
$$
\ff(\hat{\vv}(T))\not\equiv 0 \mod \hat{q}(T)
$$
then there is no exact RUR of a rational component of $\I$ within the distance of $E$ from $q(T), v_1(T), \ldots, v_n(T)$.
\end{theorem}

\begin{proof} If there was an exact RUR of a rational component of $\I$ within $E$ from $q(T), v_1(T), \ldots, v_n(T)$, the heights of its coefficients would be bounded by $B\geq Hn\delta e^{12\delta n^4 D^{n+1}}$ as in (\ref{Hbound2}). The rational number reconstruction algorithm would have found this exact RUR, 
which is a contradiction.
\end{proof}

Finally, we give the number of iteration needed (asymptotically) in the ``best case" and in the ``worst case". 

\begin{theorem} Let $\ff \in\Q[x_1, \ldots, x_n]^m $ and $\I=\langle \ff\rangle$ be as above.
\begin{enumerate}
\item Assume that  $q^*(T), v_1^*(T), \ldots, v_n^*(T)$ is an exact RUR for a component of $\I$ and 
assume that  $q^{(0)}(T), v_1^{(0)}(T), \ldots, v_n^{(0)}(T)$ is  an initial approximate RUR which quadratically converges to $q^*(T), v_1^*(T), \ldots, v_n^*(T)$ either using local Newton iteration as in Subsection \ref{local} or global Newton iteration as in Subsection \ref{global}. Then the number  of iterations needed to find  $q^*(T), v_1^*(T), \ldots, v_n^*(T)$ is asymptotically bounded by 
$$
{\mathcal O} ( \log(d)\log\log (H^*E_0))
$$
where $
H^*=max\{H(q^*), H(v_1^*), \ldots, H(v_n^*)\}
$
is the height of the output, $d=\deg_T(q)$, and $E_0$ is an upper bound on the Euclidean distance of $q^{(0)}(T), \vv^{(0)}(T)$ from the output $q^*(T), \vv^*(T)$. 

\item Assume that $q^{(0)}(T), v_1^{(0)}(T), \ldots, v_n^{(0)}(T)$ is  an initial approximate RUR which quadratically converges to the polynomials $q^*(T), v_1^*(T), \ldots, v_n^*(T)$, but these polynomials have  irrational coefficients, i.e. they do not form an RUR of a rational component of $\I$. In this case we need 
$$
{\mathcal O}  (n\log(d\delta n D) \log\log(HE_0))
$$ iterations to conclude that our  iteration did not converge to an exact RUR of a rational component of $\I$. Here $H$ and $D$ are the maximum of the heights and degrees of the polynomials in $\ff$, respectively,  $\delta$ is the number of roots in $V(\I)$, and $d, E_0$ are as above in 1.
\end{enumerate}
\end{theorem}

\begin{proof} 1. By Theorem \ref{bestcase}  to successfully terminate the algorithm with the exact RUR we need to achieve inaccuracy of $\frac{1}{2(H^*)^2}$. We will use the error bound of the local Newton method  to increase the accuracy of the RUR given in (\ref{errorbound}) because that is a weaker bound than the one we got for the global Newton method in (\ref{errorbound2}), so the bound we get for the number of iterations $k$ in (\ref{errorbound}) will also give an upper bound for the iterations that we need using the global Newton method. 
Thus using  (\ref{errorbound})  we need that 
$$
d^2E_0\left(\frac{1}{2}\right)^{2^k-1-d}\leq \frac{1}{2(H^*)^2},
$$ which is satisfied if 
$$
k\geq c_1 \log (d)\log\log(E_0H^*)
$$ for some constant $c_1\leq 2$.\\
2. By Theorem \ref{worstcase}  to terminate the algorithm  in the worst case we need to achieve accuracy
$$
\frac{1}{2(Hn\delta e^{12\delta n^4 D^{n+1}})^2}.
$$
Using the same argument as for the first claim for the error bound after $k$ iterations of either the local or the global Newton method, using  (\ref{errorbound})  we need that 
$$
d^2E_0\left(\frac{1}{2}\right)^{2^k-1-d}\leq \frac{1}{2(Hn\delta e^{12\delta n^4 D^{n+1}})^2},
$$ which is satisfied if 
$$
k\geq c_2 (n+1)\log(d\delta n D) \log\log(E_0H)
$$
for some constant $c_2\leq 18$.
\end{proof}

\section{Certification of isosingular points}\label{Sec:MultipleRoots}

As summarized in the Introduction, we will use 
isosingular deflation \cite{Hauenstein-Wampler}
to reduce to the case of certifying simple roots 
to overdetermined systems which was 
discussed in Section~\ref{Sec:Over}.
Due to this reduction, we can extend this approach to 
all points which can be regularized by isosingular deflation,
called {\em isosingular points}.
Since every isolated multiple root is an isosingular point, this method
applies to multiple roots.
However, isosingular points need not be isolated
as demonstrated by the origin with respect to the
Whitney umbrella $x^2 - y^2 z = 0$.

\subsection{Isosingular deflation}\label{SSec:Iso}

Given a system $\g=(g_1,\ldots,g_m)\in \Q[x_1,\ldots,x_n]$
and a root $z$ of $\g = 0$,
isosingular deflation 
is an iterative regularization process.
For the current purposes, we will only consider
the deflation operator $\sD_{\det}$ defined as follows.  
Let $d = \dnull(\g,z) = \dim~\Null~J\g(z)$
be the dimension of the null space of the Jacobian
matrix $J\g$ evaluated at $z$.
Let $\ell = \binom{n}{n-d+1}\cdot\binom{m}{n-d+1}$
and $\{\sigma_1,\dots,\sigma_\ell\}$ be the index
set of all $(n-d+1)\times(n-d+1)$ submatrices 
of an $m\times n$ matrix.  If $d = \max\{0,n-m\}$, 
then $\ell=0$ in which case we know that 
$z$ is a smooth point on an irreducible solution component of 
dimension~$d$.  Define $\sD_{\det}(\g,z) = (\g_{\det},z_{\det})$
where $z_{\det} = z$ and 
$$\g_{\det} = \left[\begin{array}{c}
\g \\
\det J_{\sigma_1} \g \\
\vdots \\
\det J_{\sigma_\ell} \g
\end{array}\right].$$
The matrix $J_\sigma \g$ is the submatrix of
the Jacobian $J\g$ indexed by $\sigma$.
Clearly, $\g_{\det}$ consists of polynomials in $\Q[x_1,\dots,x_n]$.  

Since the deflation operator $\sD_{\det}$ will be 
repeatedly applied, we write $\sD_{\det}^k(\g,z)$ to mean $k$
successive iterations with $\sD_{\det}^0(\g,z) = (\g,z)$.
This naturally leads to an associated nonincreasing sequence
of nonnegative integers, called
the {\em deflation sequence}, namely $\{d_r(\g,z)\}_{r=0}^\infty$
where
$$d_k(\g,z) = \dnull(\sD_{\det}^k(\g,z)) \hbox{~~~~for $k\geq0$}.$$
Since this  nonnegative sequence is decreasing, it must have a limit and it must stabilize to that limit.
That is, there are integers $d_\infty(\g,z)\geq0$ and $s\geq0$ 
so that $d_t(\g,z) = d_\infty(\g,z)$ for all $t\geq s$.
When $z$ is isolated, $s$ is bounded above by the depth as well as multiplicity
\cite{DZ05,Hauenstein-Wampler,LVZ06}.
The limit $d_{\infty}(\g,z)$ is called the {\em isosingular local dimension}
of $z$ with respect to $\g$.  The {\em isosingular points} are
those for which their isosingular local dimension is zero
so that, after finitely many iterations, 
isosingular deflation has regularized the root, i.e., constructed
a polynomial system for which the point is a regular root.
Clearly, such a system must consist of at least $n$ polynomials, but
will typically be overdetermined.  

Since the resulting polynomial systems have rational coefficients,
it immediately follows that every point in a zero-dimensional rational component 
must have the same deflation sequence.  This can be used
to partition the set of points under consideration 
into subsets and run the certification procedure
described in Section~\ref{Sec:Over} independently on each subset.  

The construction of the deflation sequence and the resulting regularized
system is an exact process that depends upon $z$.  In our situation
where $z$ is only known approximately, we use the numerical approximations
to compute exact numbers, namely the nonnegative integers 
arising as the dimensions of various linear subspaces which form the deflation sequence.  

One drawback with $\sD_{\det}$ is the number of minors used in each iteration, 
namely $\ell = \binom{n}{n-d+1}\cdot\binom{m}{n-d+1}$.
Since the codimension of the set of $m\times n$ matrices of rank $n-d$
is $c = d(m+d-n)$, we will adjust $\sD_{\det}$ to use exactly $c$ minors as follows.
Since $d = \dnull(\g,z)$, we can select an invertible 
$(n-d)\times(n-d)$ submatrix of $J\g(z)$.  
Rather than using
all of $\{\sigma_1,\dots,\sigma_\ell\}$, we only use the 
$c$ many which contain our selected invertible submatrix.  
In particular, with this setup, 
the tangent space of these $c$ many minors
is equal to the tangent space of all $\ell$ minors at $z$.

With this specialized construction, one now needs to be cautious
that two points with the same deflation sequence can fail to be
regularized by the system constructed by the other.  However,
all points in the same zero-dimensional rational component
will still be regularized simultaneously.
In particular, by comparing ranks of various submatrices, 
one may be able to produce a finer partition of the points
under consideration before independently certifying each collection.

% One may add fewer than $c$ as long as the resulting tangent space of the new system has the proper dimension.  Locally, one can show that this is enough to yield the proper deflation.

\begin{example}\label{Ex:Whitney}
As a demonstration, consider the Whitney umbrella 
defined by $\g(x,y,z) = x^2 - y^2z$.  
Following \cite[Ex.~5.12]{Hauenstein-Wampler}, the deflation sequence
for the origin is $\{3,2,0,0,\dots\}$ showing that the origin is 
not isolated but is an isosingular point.  In particular, it
takes two iterations to construct a polynomial system for which the origin is a regular root.  Since $J\g(0)$ is identically zero, the first iteration
appends all partial derivatives,~say
$$\g'(x,y,z) = \left[\begin{array}{c} \g(x,y,z) \\ x \\ yz \\ y^2 \end{array}\right].$$
Since $J\g'(0)$ has rank 1, the original formulation of $\sD_{\det}$ 
will append $18$ $2\times2$~minors of $Jg'$.  
However, with our modification, we only need to add the $6$ minors
which arise by submatrices that, in this case, 
include the unique nonzero element of $Jg'(0)$.
For this example, it is easy to verify that the ideal of the 
resulting regularizing polynomial system is equal to $\langle x,y,z\rangle$.
\end{example}

\subsection{Certification}\label{SSec:SingCert}

Given $\ff=(f_1, \ldots, f_m)\in \Q[x_1, \ldots, x_n]$ 
and a subset $V\subset V(\ff)$ consisting of isosingular points,
the process for certification proceeds as follows.  

\begin{description}
\item[1. Deflation sequences.]  Compute the deflation sequences 
for each of the points in $V$.  If each point is an isosingular
point, then isosingular deflation will terminate and produce
a regularized system for each point.  If one is not an isosingular
point, one can apply the tests developed in \cite[Section~6]{Hauenstein-Wampler} 
to determine that the sequence has stabilized with the point having
a positive isosingular local dimension.  Remove all such points from $V$
and partition the remaining points based on their deflation sequences
and common regularizing polynomial systems, say $V_1,\dots,V_k$.  

\item[2. Certify each $V_i$.] Associated with each $V_i$ is a
polynomial system $\ff^{(i)}$ having rational coefficients that must
be either well-constrained or overdetermined.  If it is well-constrained,
simply apply standard $\alpha$-theoretic techniques for certification.
If overdetermined, use the approach presented in Section~\ref{Sec:Over}
for certification.  

\end{description}

Successfully completing the certification proves that the points
under consideration are indeed isosingular points of $\ff$,
i.e., the isosingular local dimension is zero.  
However, as currently formulated, this does not yield any information 
about the local dimension of the points, e.g., 
deciding if the point is isolated or not.  Furthermore,
even if one knows that a given point is isolated, this approach 
currently does not yield information about its multiplicity.
Both of these are topics that will be addressed in future~work.

\section{Examples}\label{Sec:Ex}

\subsection{An illustrative example}\label{SSec:Illus}

To demonstrate the approach, consider the polynomial system
$$\g(x_1,x_2,x_3) = \left[\begin{array}{c} x_1^2 + x_2^2 - 1 \\ 8x_1 - 16x_2^2 + 17 \\ x_1 - x_2^2 - x_3 - 1
\end{array}\right].$$
It is easy to verify that $\g$ has two roots of multiplicity $2$.  
Thus, after appending $\det J\g$, we are interested in 
the overdetermined polynomial system 
$$\ff(x_1,x_2,x_3) = \left[\begin{array}{c} \g(x_1,x_2,x_3) \\ 64x_1x_2 + 16x_2
\end{array}\right]$$
which has two regular roots.  A randomization of $\ff$
consists of $3$ quadratics which has $4$ regular solutions,
two of which can be  shown to not correspond to roots of $\ff$
via \cite[Section 3]{Hauenstein-Sotille}.
We start with the following numerical approximations for the
$d = 2$ points of interest:
$$\bfz_1 = (-0.250,0.968,-2.188) \hbox{~~and~~} \bfz_2 = (-0.250,-0.968,-2.188)$$
with error bound $\varepsilon = 0.002$.  From these numerical approximations, we
see that we can take the primitive element to be $u = x_2$.  

Using exact arithmetic, the initial RUR corresponding to this setup is
$$q(T) = T^2 - 14641/15625, ~~v_1(T) = -1/4, ~~v_2(T) = T, ~~v_3(T) = -547/250.$$
At $k = 0$ with $E = \varepsilon d^2 2^{1+d-2^k}$, we have $B = \left\lceil (2E)^{-1/2}\right\rceil = 2$.  Since $1/4$ can not be approximated by a rational
number with denominator bounded by $B = 2$ with an error of at most 
$(2B^2)^{-1} = 1/8$, we perform a Newton iteration on each $\bfz_i$.
At $k = 1$, using a denominator bound of $B = 16$ with error tolerance $1/512$,
we obtain
$$q(T) = T^2 - 15/16, ~~v_1(T) = -1/4, ~~v_2(T) = T, ~~v_3(T) = -35/16.$$
Since
$$\ff(-1/4,T,-35/16) = \left[\begin{array}{c} 
T^2 - 15/16 \\ 2(T^2-15/16) \\ -(T^2-15/16) \\ 0
\end{array}\right] = 0 ~~ \mod q(T)$$
we have proven that $\ff$ and $\g$ have (at least) $2$ roots which form
a rational component.
The corresponding well-constrained system from this RUR is
$$\left[\begin{array}{c} x_1 + 1/4 \\ x_2^2 - 15/16 \\ x_3 + 35/16 \end{array}\right].$$

\subsection{A well-constrained system with regular and multiple roots}\label{SSec:Caprasse}

A common benchmark system is the Caprasse system which has $24$ regular
roots and $8$ roots of multiplicity four \cite{MK99}, namely 
$$\g = 
{\small
\left[\begin{array}{c}
x_1^3 x_3-4 x_1^2 x_2 x_4-4 x_1 x_2^2 x_3-2 x_2^3 x_4-4 x_1^2-4 x_1 x_3+10 x_2^2+10 x_2 x_4-2 \\
x_1 x_3^3-4 x_1 x_3 x_4^2-4 x_2 x_3^2 x_4-2 x_2 x_4^3-4 x_1 x_3+10 x_2 x_4-4 x_3^2+10 x_4^2-2 \\
2 x_1 x_2 x_4+x_2^2 x_3-2 x_1-x_3 \\
x_1 x_4^2+2 x_2 x_3 x_4-x_1-2 x_3 
\end{array}\right].}$$
Since the system is well-constrained, numerical approximations
for the $24$ regular roots can be certified using standard $\alpha$-theory.
Here, we consider certifying the multiple roots.
At each of these multiple roots, $J\g$ has rank~$2$
with the lower right $2\times 2$ block having full rank.
Thus, we consider the 
system $\ff$ constructed by appending the four 
$3\times 3$ minors of $J\g$ containing the lower right block to $\g$.

From the numerical approximations of the $8$ points $\bfz_i$
that we computed using {\tt Bertini} \cite{Bertini}, we see
that $u = x_1 - x_2 + 3x_3 - 3x_4$ is a primitive element.
Starting the numerical approximations correct to $10$ digits, we 
obtain the following RUR after one Newton iteration:
$${\small \begin{array}{rcl}
q(T) &=& (T^2 + 4/3)(T^2 + 12)(T^2 - 16 T + 76)(T^2 + 16 T + 76) \\
v_1(T) &=& -(4107 T^7 - 347060 T^5 + 15954064 T^3 + 361834048 T)/339935232 \\
v_2(T) &=& -(6999 T^7 - 672196 T^5 + 32397008 T^3 + 384342848 T)/679870464 \\
v_3(T) &=& (1851 T^7 - 169876 T^5 + 8058512 T^3 + 181018688 T)/169967616 \\
v_4(T) &=& (6999 T^7 - 672196 T^5 + 32397008 T^3 + 384342848 T)/679870464
\end{array}}$$
This RUR shows that the $8$ roots arise from $4$ rational components,
each of degree $2$, with splitting field $\Q[\sqrt{3}]$.

\subsection{Two cyclic system}\label{SSec:Cyclic}

A common family of benchmark systems are the cyclic-$n$
roots \cite{CyclicN}.  Here, we demonstrate
using the approach for problems related to $n = 4$ and $n = 9$.

For $n\geq2$, the cyclic-$n$ system is
$$\ff_n = \left[\begin{array}{cc}
\sum_{j=1}^n \prod_{k=1}^{D} x_{j+k} & \hbox{~~for $D = 1,\dots,n-1$} \\
\prod_{k=1}^n x_k - 1 
\end{array}\right]$$
where $x_{n+\ell} = x_\ell$ for all $\ell = 1,\dots,n$.

The solutions of $\ff_4 = 0$ lie on two irreducible curves,
with $8$ embedded points, which are isosingular points.
We deflate these points simultaneously by appending in
the four $3\times3$ minors of $J\ff_4$ containing
the first and last rows, and second and third columns. 
By using numerical approximations computed by {\tt Bertini},
we see that we can use the primitive element
\mbox{$u = x_1 + 2x_2 - x_3 + 3x_4$}.  This yields
the RUR:
$$
\begin{array}{rcl}
q(T) &=& (T-1)(T+1)(T-3)(T+3)(T^2+1)(T^2+9) \\
v_1(T) &=& (-T^5 + 121T)/120 \\
v_2(T) &=& (-T^5+61T)/60 \\
v_3(T) &=& (T^5-121T)/120 \\
v_4(T) &=& (T^5-61T)/60 
\end{array}
$$

Now, for $n = 9$, we consider the overdetermined system
$$\ff(x) = \left[\begin{array}{c} \ff_9(x) \\ x_1 - x_4 \\ x_1 - x_7 \end{array}\right]$$
motivated by Example 9 of \cite[Section~7]{DZ05}.  
The degree of the ideal generated by $\ff$ is $162$ 
with {\tt Bertini} computing $54$ regular points and $54$ double points.
The following uses the primitive element
$u = x_1 + 2 x_2 - x_3 + 2 x_5 + x_6 - x_8$.

For the $54$ regular points, we compute
$$
\begin{array}{rcl}
q(T) &=& (T^2+T-101)(T^4-T^3+102 T^2+101 T+10201) \\
& & (T^2+19 T+79)(T^4-19 T^3+282 T^2-1501 T+6241) \\
& & (T^2-17 T+61)(T^4+17 T^3+228 T^2+1037 T+3721) \\
& & (T^{12}-2356 T^9+5057697 T^6-1161599884 T^3+243087455521) \\
& & (T^{12}-304 T^9+1122717 T^6+313211504 T^3+1061520150601) \\
& & (T^{12}+1802 T^9+3020223 T^6+409019762 T^3+51520374361). 
\end{array}
$$

For the $54$ double points, by comparing ranks of $8\times8$ submatrices
of $J\ff$, we are able to partition into $3$ subsets of size $18$.  
For each of these subsets, we compute the three as
$$
\begin{array}{rcl}
q(T) &=& (T^2-11 T+19) (T^4+11 T^3+102 T^2+209 T+361)\\
& & (T^{12}+704 T^9+488757 T^6+4828736 T^3+47045881)
\\
\\
q(T) &=& (T^2+13 T+31) (T^4-13 T^3+138 T^2-403 T+961) \\
& & (T^{12}-988 T^9+946353 T^6-29433508 T^3+887503681) 
\\
\\
q(T) &=& (T^2+T-11)(T^4-T^3+12 T^2+11 T+121)\\
& & (T^{12}-34 T^9+2487 T^6+45254 T^3+1771561).
\end{array}
$$

\section*{Acknowledgements}
We would like to thank Eric Schost and Teresa Krick for helpful discussions on {\em a priori} upper bounds for the heights of RUR's, as well as Hoon Hong for his helpful suggestions.

\section*{References}
\bibliographystyle{elsarticle-num}
%\bibliography{bibl}
%\input{CertifiedBibl.bbl}
\def\cprime{$'$}

\end{document}